\documentclass[11pt]{article}
\usepackage{fullpage}
\usepackage[T1]{fontenc}
\usepackage[utf8]{inputenc} 
\usepackage{xcolor}

\usepackage{amsmath, amssymb, amsthm, amsfonts, graphicx, color, subcaption, enumerate, bm, array, mathtools}

\usepackage[colorlinks=true]{hyperref}
\definecolor{red}{HTML}{F88379}
\hypersetup{
    linkcolor=blue
    ,citecolor=orange
    ,filecolor=purple
    ,urlcolor=orange
    ,menucolor=purple
    ,runcolor=purple
}
\usepackage[capitalise,nameinlink]{cleveref}

\usepackage[colorinlistoftodos,prependcaption,textsize=tiny]{todonotes}

\usepackage{multicol, multirow}
\usepackage{thmtools, thm-restate}
\usepackage{ragged2e}
\usepackage{lineno}
\usepackage{framed}
\usepackage[framemethod=tikz]{mdframed}

\usepackage[linesnumbered, ruled, vlined]{algorithm2e}

\usepackage{subfiles}
\graphicspath{{graphics/}}

\crefname{claim}{Claim}{Claims}
\crefname{property}{Property}{Properties}
\crefname{algocf}{Algorithm}{Algorithms}
\Crefname{algocf}{Algorithm}{Algorithms}

\usepackage[style=trad-alpha,natbib=true,maxcitenames=4]{biblatex}
\usepackage{dirtytalk}

\makeatletter
\g@addto@macro\bfseries{\boldmath}
\makeatother

\newtheorem{lemma}{Lemma}[section]

\newcommand{\local}{\mathsf{LOCAL}}

\newcommand{\congest}{\mathsf{CONGEST}}
\newcommand{\mpc}{\mathsf{MPC}}

\newcommand{\eps}{\varepsilon}
\renewcommand{\epsilon}{\varepsilon}
\newcommand{\poly}{\operatorname{poly}}

\newcommand{\E}{\mathbb{E}}

\DeclareMathOperator{\dist}{dist}

\newcommand{\bp}{\mathsf{bp}}

\newcommand{\dtmax}{T_{\operatorname{DMM}}}
\newcommand{\rtmax}{T_{\operatorname{RMM}}}

\newcommand{\AMM}{\mathsf{AlmostMaximalMatching}}
\newcommand{\GPR}{\mathsf{GuardedProposalRound}}
\newcommand{\GQM}{\mathsf{GuardedQuantileMatch}}
\newcommand{\DGASM}{\mathsf{GuardedASM}}

\title{A Few Shared Random Bits Suffice \\ for Constant-Round Almost Stable Matching}

\author{Yi-Jun Chang\footnote{National University of Singapore. ORCID: 0000-0002-0109-2432. Email: cyijun@nus.edu.sg} \and Kushagra Chatterjee\footnote{Indian Statistical Institute, Kolkata. ORCID: 0009-0008-6838-2216. Email: kushagrachatterjee@gmail.com. The work was done when the author was a PhD student at the National University of Singapore.}}

\date{}

\begin{document}

\maketitle
\thispagestyle{empty}

\begin{abstract}
We show that almost stable matching can be solved in constant distributed
rounds on general bipartite graphs $G=(V,E)$ using only a few shared random bits.
Specifically, in the $\congest$ model, we compute a matching whose expected
number of blocking pairs is at most $\varepsilon |E|$ in
$O\left(\frac{\log(1/\varepsilon)}{\varepsilon^4}\right)$ rounds using
$O\left(\log(1/\varepsilon)\right)$ shared random bits.
Thus, for every constant $\varepsilon>0$, the round complexity is $O(1)$,
independent of the number of vertices and the maximum degree.

Previous algorithms achieve constant round complexity only for bounded-degree
or almost-regular graphs; on general graphs, their round complexity depends
polylogarithmically on $n$.
Our main technical idea is a degree-guarded freezing rule that allows
widely varying degrees to be handled by a single global charging
argument, avoiding the $\Theta(\log n)$ successive degree thresholds used
in previous work.
The shared random bits are used only to select a common random output
iteration.

As consequences, we obtain an
$O\left(
    \frac{\log(1/\varepsilon)}{\varepsilon^4}
    +
    \frac{\log n}{\varepsilon}
\right)$-round $\congest$ algorithm without pre-shared randomness, via a
low-diameter decomposition, and an
$O\left(\frac{\log(1/\varepsilon)}{\varepsilon^4}\right)$-round algorithm
in the fully-scalable Massively Parallel Computation ($\mpc$) model with
linear total memory.
\end{abstract}

\newpage
\bigskip
\tableofcontents
\bigskip
\thispagestyle{empty}

\newpage
\pagenumbering{arabic}

\section{Introduction}

The \emph{stable marriage} problem, introduced by Gale and Shapley~\cite{gale1962college}, is a classical problem at the intersection of algorithms, economics, and matching theory.
An instance consists of a bipartite graph $G=(Y\cup X,E)$ with $n=|Y\cup X|$ vertices, where the vertices in $Y$ and $X$ represent men and women, respectively, and each vertex has a strict preference ordering over its neighbors.
A matching $M$ is \emph{stable} if it has no \emph{blocking pair}: an
edge $\{m,w\}\in E\setminus M$ such that both endpoints prefer each other
to their partners in $M$, where being unmatched is considered worse than
being matched to any neighbor.

Gale and Shapley showed that every such instance admits a stable matching via a simple proposal algorithm.
In each round, every unmatched man proposes to his most preferred woman who has not yet rejected him; each woman keeps her favorite proposal seen so far and rejects the others.
Rejected men continue down their preference lists, and the process eventually reaches a stable matching.

This proposal process has a natural distributed interpretation, where each vertex is a processor and the edges of $G$ are communication links.
In the standard $\local$ and $\congest$ models, vertices communicate synchronously with their neighbors; message sizes are unbounded in the $\local$ model and limited to $O(\log n)$ bits per edge per round in the $\congest$ model.

Despite the simplicity of the Gale--Shapley algorithm, computing an exact stable matching is inherently \emph{global}:
Flor\'{e}en, Kaski, Polishchuk, and Suomela~\cite{floreen2010almost} showed that it requires $\Omega(n)$ rounds even in the $\local$ model.
Kipnis and Patt-Shamir~\cite{kipnis2009note} further showed an $\Omega\left(\sqrt{n/\log^2 n}\right)$ lower bound in the $\congest$ model, even on graphs of diameter $O(\log n)$ and when $O(\sqrt n)$ blocking pairs are allowed.

\paragraph{Almost stable matching.} These lower bounds motivate the study of \emph{almost stable matching}.
Following~\cite{eriksson2008instability,ostrovsky2015fast}, we say that a matching is $(1-\epsilon)$-stable if it has at most $\epsilon|E|$ blocking pairs.
The central question is whether allowing a small fraction of blocking pairs makes the problem \emph{local}, in the sense that it can be solved in a number of rounds independent of $n$.

For bounded-degree graphs, the answer is yes.
\citet{floreen2010almost} showed that truncating the Gale--Shapley algorithm after $O(\Delta^2/\varepsilon)$ rounds, where $\Delta$ is the maximum degree, already produces a matching $M$ with at most $\varepsilon|M|$ blocking pairs, and hence at most $\varepsilon|E|$ blocking pairs.
Thus, for bounded-degree graphs and constant $\varepsilon$, $(1-\varepsilon)$-stable matching can be solved in constant $\congest$ rounds.

Ostrovsky and Rosenbaum~\cite{ostrovsky2015fast} were the first to obtain an efficient algorithm for general bipartite graphs with arbitrarily large and highly nonuniform degrees.
They showed that a $(1-\epsilon)$-stable matching can be computed in
$O\left(\frac{\dtmax \log n}{\epsilon^3}\right)$ rounds deterministically and
$O\left(\frac{\rtmax \log n}{\epsilon^3}\right)$ rounds with probability $1-1/\poly(n)$, where $\dtmax$ and $\rtmax$ denote the deterministic and randomized round complexities of maximal matching, respectively.

The state-of-the-art upper bounds for maximal matching in the $\congest$ model are
$\dtmax=O(\log^2\Delta\log n)=O(\log^3 n)$~\cite{fischer2020improved} and
$\rtmax=O(\log\Delta+\log^3(\log n))=O(\log n)$~\cite{barenboim2016locality}.
Thus, the Ostrovsky--Rosenbaum algorithm takes 
$O\left(\frac{\log^4 n}{\epsilon^3}\right)$ deterministic rounds and
$O\left(\frac{\log^2 n}{\epsilon^3}\right)$ randomized rounds.
For comparison, in the deterministic $\local$ model, maximal matching can be computed in
$O(\log^{5/3} n \cdot \poly \log \log n)$ rounds~\cite{ghaffari2024near}.
For randomized algorithms, Khoury and Schild~\cite{khoury2025round}, improving on the classical lower bound of Kuhn, Moscibroda, and Wattenhofer~\cite{kuhn2016local}, showed that maximal matching requires
$\Omega\left(\min\{\log\Delta,\sqrt{\log n}\}\right)$
rounds even in the $\local$ model.
Thus, on general graphs, maximal matching necessarily incurs a polylogarithmic dependence on at least one of $n$ and $\Delta$.

For the special case of almost-regular graphs, Ostrovsky and Rosenbaum~\cite{ostrovsky2015fast} also gave an algorithm whose round complexity depends only on $\epsilon$ and the failure probability, and is therefore constant when these parameters are constant.

This leaves a basic question:

\begin{center}
    \emph{Can almost stable matching be solved in constant rounds on general bipartite graphs?}
\end{center}

\subsection{New Results}

We answer the above question in the affirmative, assuming a small amount of shared randomness: only $O(\log(1/\epsilon))$ shared random bits are needed to make the round complexity independent of $n$.

For a matching $M$, let $\bp(M)$ denote the number of blocking pairs with respect to $M$.

\begin{restatable}[Almost stable matching with shared randomness]{theorem}{maincongest}
    Let $G=(Y\cup X,E)$ be a stable marriage instance.
    For every $0<\epsilon\leq 1/2$, there is a randomized algorithm in the  $\congest$ model with $O(\log(1/\epsilon))$ bits of shared randomness that terminates in
    $O\left(\frac{\log(1/\epsilon)}{\epsilon^4}\right)$ rounds and outputs a matching $M$ satisfying
    $\E[\bp(M)]\leq \epsilon|E|$.
    \label{thm:main}
\end{restatable}

Consequently, for every constant $\epsilon>0$, almost stable matching can be solved in constant $\congest$ rounds on general bipartite graphs.
By running the algorithm with parameter $\epsilon/2$ and applying Markov's inequality, we obtain a $(1-\epsilon)$-stable matching with probability at least $1/2$.

There are two sources of dependence on $n$ in the algorithm of Ostrovsky
and Rosenbaum~\cite{ostrovsky2015fast} for general bipartite graphs.
First, their proposal process repeatedly computes maximal matchings.
This dependence can be removed by relaxing maximality and using
almost-maximal matchings instead.

The second source is more fundamental.
Their algorithm handles highly nonuniform degrees through
$\Theta(\log n)$ successive degree thresholds, treating different degree
ranges one after another.
Our main technical contribution is to avoid these thresholds altogether.
We retain their quantile-based proposal framework, but introduce a
degree-guarded freezing rule that makes every displacement of a
high-degree man chargeable to the degree of the woman that displaced him.
This leads to a single global charging argument that handles all degrees
simultaneously.

The resulting bound is an average over the iterations of the algorithm,
rather than a guarantee for the final iteration.
This is where the shared randomness enters:
$O(\log(1/\epsilon))$ shared random bits suffice for all vertices to agree
on a uniformly random iteration whose matching is returned.
No other part of the algorithm requires shared randomness.

We can also remove the assumption of pre-shared randomness, at the cost of an additive $O\left(
        \frac{\log n}{\epsilon}
    \right)$ term in the round complexity.
The idea is to combine \Cref{thm:main} with the low-diameter decomposition of Miller, Peng, and Xu~\cite{miller2013parallel}. 

\begin{restatable}[Almost stable matching without shared randomness]{theorem}{decomposition}
\label{cor:decomposition}
    Let $G=(Y\cup X,E)$ be a stable marriage instance.
    For every $0<\epsilon\leq 1/2$, there is a randomized algorithm in the $\congest$ model, without pre-shared randomness, that terminates in
    $O\left(
        \frac{\log(1/\epsilon)}{\epsilon^4}
        +
        \frac{\log n}{\epsilon}
    \right)$
    rounds and outputs a matching $M$ satisfying
    $\E[\bp(M)]\leq \epsilon|E|$.
\end{restatable}

For constant $\epsilon$, this improves the randomized upper bound of Ostrovsky and Rosenbaum~\cite{ostrovsky2015fast} from
$O(\rtmax\log n)$ rounds to $O(\log n)$ rounds.
In particular, using the best known $\congest$ bound $\rtmax=O(\log n)$, the worst-case round complexity improves from $O(\log^2 n)$ to $O(\log n)$.
This improvement is not due to a weaker success guarantee: for constant $\epsilon$, their randomized bound remains $O(\log^2 n)$ even when the desired success probability is only a fixed constant.

The low-diameter decomposition partitions the graph into low-diameter clusters while cutting only a small fraction of the edges in expectation.
Within each cluster, the vertices can generate the shared random bits required by \Cref{thm:main} by broadcasting from a cluster center.
We then run the algorithm independently inside the clusters and simply count every inter-cluster edge as a potential blocking pair.

We also consider the \emph{Massively Parallel Computation} ($\mpc$) model.
In this model, a large graph is distributed among many machines with
limited local memory, and computation proceeds in synchronized rounds.
A central goal is to minimize the number of such rounds while keeping both the memory per machine and the total memory small.

We work in the \emph{fully-scalable} regime, where each machine has only $n^\delta$ words of local memory for an arbitrarily small fixed constant $\delta>0$.
This regime is particularly demanding for graph problems with high-degree vertices, since the neighborhood of a single vertex need not fit on one machine.
Nevertheless, our $\congest$ algorithm can be implemented in this model without increasing its asymptotic round complexity, and with total memory proportional to the input size.
Shared randomness incurs no additional asymptotic cost in the $\mpc$ model, since a common random seed can be generated and broadcast in $O(1)$ rounds.

\begin{restatable}[Almost stable matching in fully-scalable MPC]{theorem}{massivelyparallel}
\label{cor:mpc}
    Let $G=(Y\cup X,E)$ be a stable marriage instance.
    For every fixed constant $\delta\in(0,1)$ and every $0<\epsilon\leq 1/2$, there is a randomized
    algorithm in the fully-scalable $\mpc$ model with $n^\delta$ words of
    local memory per machine and $O(n+|E|)$ words of total memory that
    terminates in $O\left(\frac{\log(1/\epsilon)}{\epsilon^4}\right)$
    rounds and outputs a matching $M$ satisfying
    $\E[\bp(M)]\leq \epsilon|E|$.
\end{restatable}

To the best of our knowledge, \Cref{cor:mpc} is the first algorithm for almost stable matching in the $\mpc$ model.

To complement our upper bounds, the $\Omega(n)$-round lower bound proof of Flor\'{e}en, Kaski, Polishchuk, and Suomela~\cite{floreen2010almost} for exact stable matching immediately yields an $\Omega(1/\epsilon)$ lower bound for almost stable matching, even allowing shared randomness.

\begin{restatable}[Lower bound for almost stable matching]{theorem}{lowerbound}
In the $\local$ model with shared randomness, any algorithm that computes a $(1-\epsilon)$-stable matching with probability greater than $1/2$ requires $\Omega(1/\epsilon)$ rounds.
\label{thm:lb}
\end{restatable}

Indeed, their hard instance can be taken to have $n=\Theta(1/\epsilon)$ vertices and fewer than $1/\epsilon$ edges.
Hence, any $(1-\epsilon)$-stable matching on this instance has no blocking pair and is therefore exactly stable. For completeness, we include the full lower bound proof in \Cref{sec:LB}.

\subsection{Additional Related Work}

Stable matching has attracted substantial attention across distributed and local computation, under a variety of communication models, structural assumptions, and fault models.
Khanchandani and Wattenhofer~\cite{khanchandani2017distributed} considered complete bipartite graphs in which the preference lists on one side are similar.
Amira, Giladi, and Lotker~\cite{amira2010distributed} studied weighted stable matching on complete bipartite graphs, where the preference rankings are induced by edge weights, and gave an $O(\sqrt n)$-round distributed algorithm.
Hassidim, Mansour, and Vardi~\cite{hassidim2016local} studied stable matching in the framework of local computation mechanism design.
Laveau, Manoussakis, Beauquier, Bernard, Burman, Cohen, and Pilard~\cite{laveau2017self} studied self-stabilizing distributed stable matching.
More recently, Constantinescu, Dufay, Ghinea, and Wattenhofer~\cite{constantinescu2025byzantine} studied stable matching in the presence of Byzantine participants and characterized its solvability under different communication and cryptographic assumptions.

Several notions of approximate stability have been considered.
A natural approach is to measure instability by the number of blocking pairs.
Eriksson and H\"aggstr\"om~\cite{eriksson2008instability} discuss several measures of instability and argue in favor of this blocking pair measure.
Ostrovsky and Rosenbaum~\cite{ostrovsky2015fast} measure the number of blocking pairs relative to the number of edges, which is also the notion we use.
Flor\'{e}en, Kaski, Polishchuk, and Suomela~\cite{floreen2010almost} use a closely related measure, normalizing instead by the size of the matching.
Kipnis and Patt-Shamir~\cite{kipnis2009note} consider a different, rank-sensitive notion, in which a pair is considered blocking only when both endpoints improve sufficiently in their preference rankings.
Their rank-sensitive notion is incomparable with the measures based on counting the number of blocking pairs.

The \emph{stable roommates} problem extends stable marriage from bipartite graphs to general graphs, where a stable matching need not exist.
\citet{abraham2005almost} showed that, in the centralized setting, minimizing the number of blocking pairs is $\mathrm{NP}$-hard and cannot be approximated within a factor of $n^{1/2-\delta}$ for any constant $\delta>0$, unless $\mathrm{P}=\mathrm{NP}$.

\subsection{Roadmap}
In \Cref{sec:preliminary}, we introduce the models and terminology used throughout the paper.
In \Cref{sec:technical-overview}, we give a technical overview of our results and the main ideas behind the algorithms.
In \Cref{sec:constant-round-congest}, we present and analyze our constant-round algorithm for almost stable matching with shared randomness.
In \Cref{sec:removing-shared-randomness}, we remove the shared randomness assumption using a low-diameter decomposition.
In \Cref{sec:mpc}, we adapt our algorithm to the $\mpc$ model.
In \Cref{sec:conclusions}, we discuss open questions and directions for future work.
Finally, in \Cref{sec:LB}, we give a proof of the $\Omega(1/\epsilon)$ lower bound for almost stable matching.

\section{Preliminaries}
\label{sec:preliminary}

In this section, we fix the notation used throughout the paper, introduce the notion of almost-maximal matching used by our algorithm, and describe the distributed and parallel computation models that we consider.

\paragraph{Stable marriage notation.}
Recall that a stable marriage instance is a bipartite graph
$G=(Y\cup X,E)$, where $Y$ and $X$ denote the sets of men and women,
respectively, and every vertex has a strict preference ordering over its
neighbors.
Throughout the paper, we write $n=|Y\cup X|$ and $d(v)=\deg_G(v)$.

For a matching $M$, let $p_M(v)$ denote the partner of $v$ in $M$,
with $p_M(v)=\bot$ if $v$ is unmatched.
When the matching is clear from context, we simply write $p(v)$.
Recall that $\bp(M)$ denotes the number of blocking pairs with respect to
$M$; thus, $M$ is $(1-\epsilon)$-stable if
$\bp(M)\leq\epsilon|E|$.

For our upper bounds, we may assume that $\epsilon|E|\geq 1$, which implies
$1/\epsilon\leq |E|\leq n^2$.
Indeed, if $\epsilon|E|<1$, then integrality of $\bp(M)$ implies that
every $(1-\epsilon)$-stable matching is stable.
In this case, we may simply run the exact Gale--Shapley algorithm, whose
round complexity is dominated by the upper bounds stated in this paper.

For $u\in N(v)$, let $P_v(u)$ denote the rank of $u$ in the preference
list of $v$, where a smaller rank indicates a stronger preference. As usual, every neighbor is preferred to being unmatched, so we use the convention
\[
    P_v(\bot)=d(v)+1.
\]

\paragraph{Almost-maximal matching.}
For a matching $M\subseteq E$ of a graph $G=(V,E)$, an edge
$\{u,v\}\in E$ is \emph{residual} if neither endpoint is matched in $M$.
Thus, $M$ is \emph{maximal} if it has no residual edge.

Our algorithm uses a randomized matching procedure whose expected number
of residual edges is small.
This is an edge-based analogue of the vertex-based almost-maximal
matching notion used by Ostrovsky and Rosenbaum~\cite{ostrovsky2015fast},
where the error is measured by the number of vertices incident to
residual edges. Counting the number of residual edges makes sense for our setting because the
approximation guarantee is itself measured by the number of blocking
pairs.

\paragraph{Distributed models.}
We work in the standard $\local$~\cite{linial1992locality} and
$\congest$~\cite{peleg2000distributed} models, where
the input graph is also the communication network: each vertex is a
processor, and adjacent vertices communicate in synchronous rounds.
Each vertex initially knows its own preference list, which side of the
bipartition it belongs to, and an $O(\log n)$-bit unique identifier.
Local computation is unrestricted.
In the $\local$ model, message sizes are unbounded, whereas in the
$\congest$ model, each vertex can send an $O(\log n)$-bit message to
each neighbor in each round.

A randomized distributed algorithm may use private random bits generated
independently at the vertices.
When we say that an algorithm uses $b$ bits of \emph{shared randomness},
we mean that all vertices additionally have access to the same $b$
independent unbiased random bits before the execution begins.
The cost of generating these bits is not included in the round
complexity.
In contrast, an algorithm \emph{without pre-shared randomness} starts
without such a common random string.

\paragraph{Massively parallel computation.}
The Massively Parallel Computation ($\mpc$) model was introduced by
Karloff, Suri, and Vassilvitskii~\cite{karloff2010model} as a theoretical
abstraction of large-scale parallel processing frameworks such as
MapReduce~\cite{dg04}.
The input is distributed among many machines, each with $S$ words of
local memory, where one word contains $O(\log n)$ bits.
Computation proceeds in synchronous rounds.
In each round, every machine may perform arbitrary local computation and
send and receive a total of $O(S)$ words to and from the other machines.
In addition to the round complexity and local memory, we also consider
the total memory, defined as the sum of the local memories over all
machines.

The local memory $S$ plays a fundamental role in the power of the
$\mpc$ model.
The literature commonly distinguishes three regimes:
\begin{itemize}
    \item The \emph{strongly superlinear} regime, where
    $S=\Theta(n^{1+\gamma})$ for some constant $\gamma>0$.
    \item The \emph{nearly-linear} regime, where $S$ is roughly linear
    in $n$, up to polylogarithmic factors.
    \item The \emph{fully-scalable} regime, where $S=\Theta(n^\delta)$ for an
    arbitrarily small constant $\delta\in(0,1)$.
\end{itemize}

The round complexity of basic graph problems can differ substantially
across these regimes.
For example, for problems such as maximal matching, maximal independent
set, approximate maximum matching, and approximate minimum vertex cover,
constant-round algorithms are known in the strongly superlinear
regime~\cite{harvey2018greedy,lattanzi2011filtering}, while nearly-linear
memory permits $O(\log\log n)$-round
algorithms~\cite{behnezhad2023exponentially,ghaffari2018improved}.
In the fully-scalable regime, the best known general-graph bounds for
these problems are polylogarithmic: Ghaffari and
Uitto~\cite{ghaffari2019sparsifying} give
\(
    O\left(
        \sqrt{\log\Delta}\log\log\Delta
        +
        \sqrt{\log\log n}
    \right)
\)-round algorithms.
This illustrates why obtaining fast algorithms in the fully-scalable
regime can be substantially more challenging.

Compared with the classical $\mathsf{PRAM}$ model of parallel computation, the
$\mpc$ model allows substantial local computation between global
communication rounds.
Its main complexity measure is therefore the number of such rounds,
rather than the amount of local computation performed within each round, making it well suited to modeling modern large-scale parallel
data-processing systems.

\section{Technical Overview}
\label{sec:technical-overview}

In this section, we explain the main ideas behind our algorithm.
We first review the quantile-based proposal framework of Ostrovsky and
Rosenbaum~\cite{ostrovsky2015fast}, focusing on the two sources of
dependence on $n$ in their algorithm.
We then explain how our algorithm removes both dependencies.

\subsection{The Ostrovsky--Rosenbaum Algorithm}

The algorithm of Ostrovsky and Rosenbaum~\cite{ostrovsky2015fast} can be
viewed as a quantized version of the Gale--Shapley proposal process.
Fix $k=\Theta(1/\epsilon)$.
Each vertex partitions its preference list into $k$ consecutive
\emph{quantiles} of approximately equal size, ordered from most preferred
to least preferred.

\paragraph{Proposal round.}
The basic operation is a \emph{proposal round}.
Each active unmatched man has a current quantile and proposes to every
woman in this quantile who has not yet rejected him.
Initially, a man starts with his first quantile; as the algorithm
progresses, he considers his most-preferred quantile that still contains
a possible partner.
A woman may receive proposals from several different quantiles of her own
preference list; she accepts all proposals belonging to the best such
quantile.

The accepted proposal edges form a bipartite graph $H$, which need not be
a matching.
Ostrovsky and Rosenbaum compute a maximal matching $N$ in $H$ to decide
which of these accepted proposals become matches.
Each woman incident to $N$ takes her partner in $N$, possibly displacing
her previous partner, while a woman not incident to $N$ keeps her current
partner.
A woman taking a new partner rejects all other men in the same or a worse
quantile than her new partner.
Thus, once a woman becomes matched she remains matched, and whenever she
changes partners she moves to a strictly better quantile.

\paragraph{Quantile match.}
Ostrovsky and Rosenbaum group $k$ proposal rounds into a procedure called
\emph{quantile match}.
At the beginning of the procedure, each unmatched man works on his
most-preferred nonempty quantile.
The guarantee is that, by the end of the procedure, every such man is
either matched to a woman in that quantile or has been rejected by every
woman in that quantile.

To see why $k$ proposal rounds suffice, consider a woman $w$ that receives
proposals.
If $w$ is matched to one of the men in her best quantile containing a
proposal, she rejects the other relevant men and receives no further
proposal during this execution of quantile match.
Otherwise, by maximality of $N$, all men in that best quantile whose
proposals she accepted must have been matched to other women.
Hence, if $w$ receives proposals in a later proposal round, her best
quantile containing a proposal must be strictly worse than before.
Since $w$ has at most $k$ quantiles, this can happen at most $k$ times.
Consequently, after $k$ proposal rounds, every man active at the beginning
of quantile match is either matched or has exhausted his current
quantile.

To see what repeated calls to quantile match achieve, first suppose that
all men participate.
Call a man \emph{good} if he is matched or has been rejected by every
woman in his original preference list, and \emph{bad} otherwise. 
Suppose that after $\ell$ calls to quantile match, $b$ men remain bad.
Although an individual matched man may later be displaced, the number of
good men cannot decrease throughout the process.
Indeed, the number of matched men cannot decrease, while a man rejected
by every woman in his original preference list can never become bad.
Hence, there were at least $b$ bad men in each of the $\ell$ calls.

Each bad man witnesses progress through one quantile in every call.
Across the entire process, the men can collectively exhaust at most
$k|Y|$ of their own quantiles.
Moreover, once a woman becomes matched, she remains matched, so at most
$|Y|$ women ever become matched.
Each such woman can trigger progress through at most $k$ quantiles as she
becomes matched and subsequently changes partners, giving at most another
$k|Y|$ quantile progress events.
Thus, there are at most $2k|Y|$ such events in total, and hence
\[
    b\ell \leq 2k|Y|.
\]
Since $k=\Theta(1/\epsilon)$, after
$\ell=\Theta(1/\epsilon^2)$ calls, only an $O(\epsilon)$ fraction of the
men remain bad.

\paragraph{Almost-regular graphs.}
For almost-regular graphs, this is already enough.
If the degrees of the men differ by at most a constant factor, then an
$O(\epsilon)$ fraction of bad men can be incident to only
$O(\epsilon|E|)$ edges.
Moreover, Ostrovsky and Rosenbaum replace maximal matching by a
vertex-based almost-maximal matching in this setting.
Thus, for constant $\epsilon$, constant degree ratio, and constant failure
probability, their round complexity is independent of $n$.

\paragraph{Arbitrary degrees.}
When the degrees vary arbitrarily, a small fraction of bad men is no
longer sufficient: a few bad men may have very large degrees and account
for many blocking pairs.
Ostrovsky and Rosenbaum overcome this issue by processing the men using
geometrically increasing thresholds on their number of remaining possible
partners.

More precisely, the algorithm has $\Theta(\log n)$ iterations.
In iteration $i$, men with at least $2^i$ remaining possible partners
participate in repeated calls to quantile match.
After sufficiently many calls, only a small fraction of these men remain
bad.
The algorithm then proceeds to iteration $i+1$, where the threshold is
doubled to $2^{i+1}$.
Any bad man left behind after iteration $i$ therefore has fewer than
$2^{i+1}$ remaining possible partners, whereas the men considered in
iteration $i$ started with at least $2^i$.
Thus, up to a constant factor, a bad man cannot contribute more blocking
pairs than a good man from the same iteration has remaining incident
edges.
Since the bad men form only a small fraction, their contribution can be
charged to the much more numerous good men. 

\subsection{Our Algorithm}

The preceding discussion reveals two sources of dependence on $n$ in the
Ostrovsky--Rosenbaum algorithm.
First, every proposal round computes a maximal matching.
Second, handling arbitrary degrees requires $\Theta(\log n)$ iterations
with geometrically increasing thresholds.

Our algorithm retains the quantile-based proposal framework of
Ostrovsky and Rosenbaum, but changes how the two sources of dependence on
$n$ are handled.
The dependence coming from maximal matching can be removed by relaxing
the maximality requirement.
The more substantial issue is the $\Theta(\log n)$ sequence of degree
thresholds.
Our main idea is a degree guard that makes it possible to replace these
thresholds by a single global charging argument that handles all degrees
simultaneously.

\paragraph{Replacing maximal matching.}
In each proposal graph $H$, instead of computing a maximal matching, we
compute an edge-based almost-maximal matching.
By truncating a simple randomized maximal matching procedure for
bipartite graphs, we can make the expected fraction of residual edges at
most $\rho$ in $O(\log(1/\rho))$ rounds.
Our algorithm makes only $O(1/\epsilon^4)$ such calls, and every proposal
graph has at most $|E|$ edges.
Thus, setting $\rho=\Theta(\epsilon^5)$ ensures that only
$O(\epsilon|E|)$ residual edges are deleted in expectation.

\paragraph{Degree guard.}
The key obstacle to removing the degree thresholds is displacement.
Without the thresholds, a high-degree man may remain matched for a long
time and be displaced only near the end of the algorithm.
Such a displacement can make him unmatched while leaving live edges,
without exhausting a quantile.
This can cause the total degree of the unmatched men with live edges to
increase sharply.

Our main idea is a simple \emph{degree guard}.
Set $R=\Theta(1/\epsilon)$.
Whenever a new pair $(m,w)$ is formed with
$d(m)>R d(w)$, we \emph{freeze} the pair by deleting all other live
edges incident to $w$.
Since proposals are made only through live edges, the pair can never
subsequently be changed.
Thus, a man whose degree is much larger than his partner's can never
later become unmatched through displacement.
On the other hand, whenever a man $m$ is displaced by a woman $w$, we
necessarily have $d(m)\leq R d(w)$.
This gives exactly the degree comparison needed to charge the degree of
a displaced man to the woman who displaced him.

Freezing itself introduces little error.
For every frozen pair $(m,w)$, the degree of $w$ is less than
$d(m)/R$.
Since the frozen pairs form a matching, their male endpoints are
distinct, and their degrees sum to at most $|E|$.
Hence, the total degree of all frozen women is at most
$|E|/R=O(\epsilon|E|)$.
We can therefore pessimistically count every edge incident to a frozen
woman as a potential blocking pair.

\paragraph{A random iteration.}
With the degree guard in place, we can remove the degree thresholds
entirely and simply run quantile match repeatedly.
We refer to each call to quantile match as an \emph{iteration}.
The remaining difficulty is that displacement can cause a high-degree man
to become unmatched while retaining live edges, so the total degree of
the unmatched men with live edges can increase sharply.

We overcome this by showing that this sum of degrees is small on average
over the iterations.
After iteration $t$, let $B_t$ be the set of unmatched men that still
have at least one live edge, and let
$P_t=\sum_{m\in B_t}d(m)$.
Consider a man $m\in B_t$.
There are two possibilities.

If $m$ was already unmatched at the beginning of iteration $t$, then
quantile match exhausts his current quantile.
Because $m$ still has a live edge afterward, this is not his final
original quantile, and hence it contains $\Omega(d(m)/k)$ original edges.
Thus, its exhaustion can pay for $d(m)$ with a factor $k$.
As every original quantile is exhausted at most once, these charges sum
to at most $k|E|$ over all iterations.

Otherwise, $m$ was matched at the beginning of iteration $t$ and became
unmatched when some woman $w$ displaced him.
The degree guard then guarantees $d(m)\leq R d(w)$.
Each woman changes partners at most $k$ times, so all such displacement
charges sum to at most $kR|E|$.
In other words, every contribution to $P_t$ is
charged either to a permanently exhausted quantile or to an improvement
of a woman's partner.
Combining the two cases,
\[
    \sum_t P_t
    \leq k|E|+kR|E|
    =k(R+1)|E|,
\]
where $k=\Theta(1/\epsilon)$ and $R=\Theta(1/\epsilon)$.

We run $\Theta(1/\epsilon^3)$ iterations.
It follows that the average value of $P_t$ over these iterations is
$O(\epsilon|E|)$.
However, the algorithm does not know which iteration has a small value
of $P_t$.
This is where the shared randomness is used.
With only $O(\log(1/\epsilon))$ shared random bits, all vertices choose
the same uniformly random iteration and output the matching saved after
that iteration.
The expected value of $P_t$ for the selected iteration is $O(\epsilon|E|)$.

Thus, rather than handling different degree ranges through
$\Theta(\log n)$ successive thresholds, the degree guard allows all
degrees to be handled simultaneously by a single global charging argument.
The shared randomness is then needed only to select the output iteration.

\section{Almost Stable Matching With Shared Randomness}
\label{sec:constant-round-congest}

In this section, we prove \Cref{thm:main}.
We describe the algorithm in \Cref{subsec:algorithm} and analyze its
correctness and round complexity in \Cref{subsec:analysis}.

\subsection{The Algorithm}
\label{subsec:algorithm}

Recall that $d(v)$ denotes the degree of $v$ in the original input graph
and $P_v(u)$ denotes the rank of $u$ in the preference list of $v$.
All degrees and quantiles below are defined with respect to the original
input graph and remain fixed throughout the algorithm.

Set
\begin{equation}
\label{eq:parameters}
    k=\left\lceil\frac{8}{\eps}\right\rceil,
    \qquad
    R=\frac{4}{\eps},
    \qquad
    L=
    2^{\left\lceil
        \log_2\left(\frac{4k(R+1)}{\eps}\right)
    \right\rceil},
    \qquad
    \rho=\frac{\eps}{4kL}.
\end{equation}

For each vertex $v$, set
\[
    s_v=\left\lceil\frac{d(v)}{k}\right\rceil
\]
and partition its preference list into consecutive quantiles
\begin{equation*}
    Q_i^v=
    \left\{
        u\in N(v):
        (i-1)s_v<P_v(u)\leq \min\{is_v,d(v)\}
    \right\},
    \qquad i\in\{1,\ldots,k\}.
\end{equation*}
The nonempty quantiles partition $N(v)$, and every nonempty quantile
except possibly the last has size $s_v$.
For an edge $(u,v)$, let $q_v(u)$ denote the index of the quantile of
$v$ containing $u$.

Throughout the algorithm, every edge is either \emph{live} or
\emph{deleted}.
Initially, all edges are live, and deletion is permanent.
Edges may be deleted for several reasons.
When the degree guard does not apply and a woman obtains a new partner,
she \emph{rejects} all other men in the same or a worse quantile by
permanently deleting the corresponding edges.
Thus, rejection is one particular type of edge deletion; other deletions
arise from residual edges of the almost-maximal matching procedure and
from the degree guard.
We will use this distinction in the analysis.

The algorithm also maintains a persistent matching $M$, represented by
the partner variables $p(v)$, with $p(v)=\bot$ initially.

\paragraph{Almost-maximal matching.}
We first describe the almost-maximal matching subroutine used in every
proposal round.
We obtain it by truncating a very simple and natural randomized procedure for maximal
matching on bipartite graphs:
in each step, every free man proposes to a uniformly random free neighbor,
every free woman receiving at least one proposal accepts one arbitrarily.
Such a maximal
matching algorithm has been used previously; see, for example, \citet{grandoni2008distributed}.

\begin{algorithm}[H]
\caption{$\AMM(H,\rho)$}
\label{alg:edge-amm}
\KwIn{A bipartite graph $H=(Y_H\cup X_H,E_H)$ and a parameter
$0<\rho<1$.}
\KwOut{A matching $N\subseteq E_H$ and the residual edge set $Z_H$.}
$N\leftarrow\varnothing$ and mark every vertex free\;
$s\leftarrow\left\lceil\log_2(1/\rho)\right\rceil$\;
\For{$i=1$ \KwTo $s$}{
    Every free man having a free neighbor in $H$ chooses one such neighbor
    uniformly at random and sends her a proposal\;
    Every free woman receiving at least one proposal chooses one proposal
    arbitrarily\;
    Add all chosen proposal edges to $N$ and mark their endpoints nonfree\;
}
$Z_H\leftarrow
\{(m,w)\in E_H:m\text{ and }w\text{ are not incident to }N\}$\;
\Return{$(N,Z_H)$}\;
\end{algorithm}

\paragraph{Guarded proposal round.}
A guarded proposal round takes as input an active set $A_m$ for each man
$m$, where all edges in $A_m$ belong to a single original quantile of
$m$.
A man with $A_m\neq\varnothing$ is \emph{active}.
Every active unmatched man proposes through every edge in $A_m$ that is
still live.
Each woman receiving proposals accepts all proposals in her best quantile
containing a proposal.

Let $H$ be the graph of accepted proposal edges.
We run $\AMM(H,\rho)$, permanently delete its residual edges, and
incorporate the resulting matching $N$ into the persistent matching $M$.
For each $(m,w)\in N$, the woman $w$ takes $m$ as her new partner,
possibly displacing her old partner.
If
\begin{equation}
\label{eq:freeze-condition}
    d(m)>R\,d(w),
\end{equation}
we call $(m,w)$ a \emph{frozen pair} and delete every other live edge
incident to $w$.
Otherwise, $w$ rejects every other live edge in the same or a worse
quantile than $m$.

Consequently, whenever a woman is matched, every other live edge incident
to her belongs to a strictly better quantile than her current partner.
In particular, every proposal subsequently received by a matched woman
comes from a strictly better quantile.

\begin{algorithm}[H]
\caption{$\GPR$}
\label{alg:guarded-proposal-round}
\KwIn{The current live edges, persistent matching $M$, and active sets
$\{A_m\}$.}
\KwOut{Updated live edges, matching, and active sets.}
Every active unmatched man $m$ proposes through every live edge in $A_m$\;
Every woman receiving at least one proposal accepts all proposals in her
best quantile containing a proposal\;
Let $H$ be the graph of accepted proposal edges\;
$(N,Z_H)\leftarrow\AMM(H,\rho)$\;
Delete every edge in $Z_H$\;
\ForEach{woman $w$ incident to an edge $(m,w)\in N$}{
    If $w$ had an old partner $m_{\mathrm{old}}$, remove the old matching
    edge and set $p(m_{\mathrm{old}})\leftarrow\bot$\;
    Set $p(w)\leftarrow m$ and $p(m)\leftarrow w$\;
    \eIf{$d(m)>R\,d(w)$}{
        Delete every live edge incident to $w$ except $(m,w)$\;
    }{
        Reject every live edge $(m',w)\neq(m,w)$ satisfying
        $q_w(m')\geq q_w(m)$\;
    }
}
Every man incident to $N$ sets $A_m\leftarrow\varnothing$\;
\end{algorithm}

\paragraph{Guarded quantile match.}
At the beginning of a guarded quantile match, every unmatched man having
a live incident edge selects his first original quantile containing a
live edge.
His active set $A_m$ is initialized to the live edges in this quantile.
The selected quantile does not change during the guarded quantile match;
edges in $A_m$ may subsequently be deleted.

We say that a man is \emph{displaced} if his partner takes a new partner
during a proposal round.
Whenever a man becomes matched, his active set is set to
$\varnothing$.
Hence, if he is later displaced during the same guarded quantile match,
he remains inactive.
Active sets are reinitialized only at the beginning of the next guarded
quantile match.

We execute $k$ guarded proposal rounds.

\begin{algorithm}[H]
\caption{$\GQM$}
\label{alg:guarded-quantile-match}
\KwIn{The current live edges and persistent matching.}
\KwOut{The updated persistent matching and edge states.}
Set $A_m\leftarrow\varnothing$ for every man $m$\;
\ForEach{unmatched man $m$ having at least one live incident edge}{
    Let $i$ be the smallest index such that $Q_i^m$ contains a live edge\;
    Set $A_m$ to the set of live edges from $m$ to $Q_i^m$\;
}
\For{$j=1$ \KwTo $k$}{
    Execute $\GPR$\;
}
\Return{the current persistent matching}\;
\end{algorithm}

Using the shared random bits, all vertices choose the same uniformly
random index $J\in\{1,\ldots,L\}$.
The algorithm performs all $L$ iterations of $\GQM$ and outputs the persistent
matching after iteration $J$.
Operationally, at the end of iteration $J$, every vertex stores its
current partner.
The choice of $J$ does not affect the execution of the proposal process.

\begin{algorithm}[H]
\caption{$\DGASM(G,\eps)$}
\label{alg:degree-guarded-asm}
\KwIn{A stable marriage instance $G=(Y\cup X,E)$ and
$0<\eps\leq 1/2$.}
\KwOut{A matching $M_J$.}

Initialize every edge as live and every vertex as unmatched\;

Using shared randomness, choose $J$ uniformly from
$\{1,\ldots,L\}$\;

\For{$t=1$ \KwTo $L$}{
    $M_t\leftarrow\GQM$\;
}
\Return{$M_J$}\;
\end{algorithm}

\subsection{Analysis of the Algorithm}
\label{subsec:analysis}

We first analyze the almost-maximal matching subroutine.

\begin{lemma}[Almost-maximal matching]
\label{lem:amm}
$\AMM(H,\rho)$ returns a matching $N$ such that
\(
    \mathbb{E}[|Z_H|]\leq \rho |E(H)|\).
\end{lemma}

\begin{proof}
Let $E_i$ denote the set of residual edges after iteration $i$ of
the algorithm, with $E_0=E(H)$.
We first show that
\begin{equation}
\label{eq:amm-shrinkage}
    \mathbb{E}[|E_i|\mid E_{i-1}]
    \leq \frac12 |E_{i-1}|.
\end{equation}

Fix the state at the beginning of iteration $i$, and let $H'$ be the subgraph
induced by the currently free vertices.
For each free man $m$, let $R_m$ be the number of residual edges
incident to $m$ after iteration $i$.
We show that
\[
    \mathbb{E}[R_m]\leq \frac{d_{H'}(m)}{2}.
\]

Fix a free man $m$ of positive degree $d=d_{H'}(m)$.
Fix the random choices of all other free men, leaving only the random
choice of $m$ unexposed.
Let $S$ be the set of neighbors of $m$ in $H'$ that receive a proposal
from at least one other man, and let $x=|S|$.

Every woman in $S$ becomes matched in iteration $i$.
Hence, any residual edge incident to $m$ must have its other endpoint
outside $S$. Therefore, if $x\geq d/2$, then we already have
\[
    R_m\leq d-x\leq \frac d2.
\]

Suppose instead that $x<d/2$.
The man $m$ chooses a woman outside $S$ with probability
\[
    \frac{d-x}{d}>\frac12.
\]
In this case, the chosen woman receives no other proposal, so $m$ becomes
matched and $R_m=0$.
With the remaining probability $x/d$, we use the trivial bound
$R_m\leq d$.
Therefore,
\[
    \mathbb{E}[R_m]
    \leq \frac{x}{d}\cdot d
    =x
    <\frac d2.
\]
Since this bound holds for every fixing of the random choices of the
other men, it also holds after averaging over those choices.

Every residual edge is incident to exactly one man.
Thus,
\[
    \mathbb{E}[|E_i|\mid E_{i-1}]
    =
    \sum_m \mathbb{E}[R_m]
    \leq
    \frac12\sum_m d_{H'}(m)
    =
    \frac12|E_{i-1}|,
\]
which proves \eqref{eq:amm-shrinkage}.

By iterated expectation,
\[
    \mathbb{E}[|E_s|]
    \leq 2^{-s}|E(H)|
    \leq \rho |E(H)|,
\]
where $s=\lceil\log_2(1/\rho)\rceil$.
\end{proof}

The next lemma is the guarded analogue of the quantile match guarantee
of Ostrovsky and Rosenbaum~\cite{ostrovsky2015fast}.
It also shows that a man matched during a guarded quantile match cannot
be displaced before the guarded quantile match ends.

\begin{lemma}[Guarded quantile match guarantee]
\label{lem:quantile-guarantee}
The following properties hold during an execution of $\GQM$.
\begin{itemize}
    \item A man who becomes matched during the execution remains matched
    until the execution ends.

    \item Every man who is active at the beginning of the execution is,
    at the end, either matched or has no live edge remaining in the
    quantile he selected at the beginning of the execution.
\end{itemize}
\end{lemma}

\begin{proof}
We first establish two facts about the proposals received by a woman $w$
during a guarded quantile match.

\paragraph{Fact 1.}
Suppose during some call to
$\GPR$, $w$ accepts proposals from quantile $i$ but does not obtain a new partner in that call.
Then, until $w$ obtains a new partner, she can only accept proposals from
quantiles with index strictly larger than $i$.

Indeed, since $i$ is the best quantile containing a proposal in that
call, $w$ receives no proposal from a quantile with index smaller than
$i$.
As active sets only shrink and live edges are only deleted during
$\GQM$, no such proposal can appear in a later call to $\GPR$.

Moreover, $w$ accepts every proposal from quantile $i$.
For each accepted edge $(m,w)$, either $m$ is incident to $N$ through
another woman, in which case $m$ becomes inactive, or neither endpoint
of $(m,w)$ is incident to $N$, in which case $(m,w)$ is residual and is
deleted.
Thus, no active man can propose to $w$ again from quantile $i$.
The claim follows.

\paragraph{Fact 2.}
Once $w$ obtains a new partner during $\GQM$, she receives no further
proposals during the same guarded quantile match.

Suppose that her new partner $m$ belongs to quantile $i$.
Since $i$ is the best quantile containing a proposal in that call to
$\GPR$, there is no proposal from a quantile with index smaller than
$i$.
As active sets only shrink and live edges are only deleted, no such
proposal can appear later.
If $(m,w)$ is frozen, all other live edges incident to $w$ are deleted.
Otherwise, all live edges in quantile $i$ or a worse quantile are
deleted.
Hence, no active man can propose to $w$ in a later call to $\GPR$ during
the same guarded quantile match.

\paragraph{Property 1.}
The first property follows immediately from Fact~2.
Once a man becomes matched, his partner receives no further proposal
during the current guarded quantile match and therefore cannot displace
him.

\paragraph{Property 2.}
Let man $m$ be active at the beginning of the execution and suppose that $m$
is unmatched at the end.
By the first property, $m$ was never matched during the execution, and
hence remains active throughout.

Suppose, for a contradiction, that some edge $(m,w)$ in the quantile
selected by $m$ at the beginning of the execution remains live at the
end.
Then $(m,w)$ belongs to $A_m$ throughout the execution, so $m$ proposes
to $w$ in each of the $k$ calls to $\GPR$.

We first claim that $w$ cannot obtain a new partner in any of these
calls.
Suppose otherwise, and let $i$ be the quantile containing her new
partner.
Since $m$ also proposes to $w$ in that call and $w$ accepts proposals
only from her best quantile containing a proposal,
\[
    i\leq q_w(m).
\]
If $m$ is incident to $N$, then $m$ becomes matched.
Otherwise, since $w$ obtains a new partner, $w$ is incident to $N$, and
the update rule deletes $(m,w)$: either the new pair is frozen, or $w$
rejects all live edges in the same or a worse quantile.
Both possibilities contradict the assumptions that $m$ remains unmatched
and $(m,w)$ remains live.

Thus, $w$ does not obtain a new partner during any of the $k$ calls to
$\GPR$.
Since $m$ proposes to $w$ in every call, $w$ accepts proposals in every
call, always from a quantile with index at most $q_w(m)$.
By Fact~1, these indices strictly increase from one call to the next.
There are at most $q_w(m)-1\leq k-1$ quantiles strictly better than
the quantile of $w$ containing $m$.
Therefore, in some call to $\GPR$, $w$ must accept proposals from the
quantile containing $m$.

In that call, $(m,w)$ belongs to the accepted proposal graph $H$.
If $m$ is incident to $N$, then $m$ becomes matched.
If $w$ is incident to $N$ but $m$ is not, then $(m,w)$ is deleted when
$w$ takes her new partner.
If neither endpoint is incident to $N$, then $(m,w)$ is residual and is
deleted.
Each case contradicts the assumptions that $m$ is unmatched and
$(m,w)$ remains live at the end.

Thus, no live edge remains in the quantile selected by $m$ at the
beginning of the execution.
\end{proof}

We next bound the total number of original edges incident to women that
ever form a frozen pair during the execution of $\DGASM$.
These edges will be accounted for pessimistically when bounding the
number of blocking pairs.
Let $X_F$ be the set of such women, and let
\begin{equation*}
    F=\{(m,w)\in E:w\in X_F\}.
\end{equation*}

\begin{lemma}[Edges incident to frozen women]
\label{lem:frozen-edges}
\(|F|\leq \frac{|E|}{R}
\).
\end{lemma}

\begin{proof}
For each $w\in X_F$, let $f(w)$ denote her partner in the frozen pair.
Since frozen pairs are permanent, the men $f(w)$ are distinct.
By the freezing condition,
\(
    d(w)<\frac{d(f(w))}{R}\), so 
\[
    |F|
    =\sum_{w\in X_F}d(w)
    <\frac1R\sum_{w\in X_F}d(f(w))
    \leq \frac1R\sum_{m\in Y}d(m)
    =\frac{|E|}{R}. \qedhere
\]
\end{proof}

Let $Z$ be the union of all residual edge sets deleted by calls to
$\AMM$ in the execution of $\DGASM$.

\begin{lemma}[Residual edge bound]
\label{lem:residual-edges} $\E[|Z|]\leq\frac{\eps}{4}|E|$.
\end{lemma}

\begin{proof}
There are exactly $kL$ calls to $\GPR$, and every accepted proposal
graph has at most $|E|$ edges.
By \Cref{lem:amm}, each call deletes at most $\rho|E|$ residual edges
in expectation.
Therefore, by linearity of expectation,
\[
    \E[|Z|]
    \leq \rho kL|E|
    =\frac{\eps}{4}|E|.
    \qedhere
\]
\end{proof}

We next bound the total degree of unmatched men that still have live
edges.
For $t\in\{1,\ldots,L\}$, recall that $M_t$ is the persistent matching after
iteration $t$ in the execution of $\DGASM$. Define
\[
    B_t=
    \left\{
        m\in Y:
        p_{M_t}(m)=\bot
        \text{ and }m\text{ has a live incident edge after iteration }t
    \right\},
    \quad
    P_t=\sum_{m\in B_t}d(m).
\]

\begin{lemma}[Degree bound for unmatched men]
\label{lem:bad-degree-charging} $\sum_{t=1}^L P_t\leq k(R+1)|E|$.
\end{lemma}

\begin{proof}
Fix an iteration $t$ and a man $m\in B_t$.

First suppose that $m$ is unmatched at the beginning of iteration $t$.
By \Cref{lem:quantile-guarantee}, the quantile selected by $m$ at the
beginning of the iteration has no live edge remaining at the end.
Since $m\in B_t$, he still has a live edge in a later quantile, so the
selected quantile is not his final original quantile and therefore has
size $s_m=\lceil d(m)/k\rceil$.
Hence $d(m)\leq k s_m$.
We charge this contribution of $d(m)$ to the selected quantile.
Each original male quantile is charged at most once, and the male
quantiles partition $E$.
Thus, the total charge of this type over all iterations is at most
$k|E|$.

Now suppose that $m$ is matched at the beginning of iteration $t$.
Since $m$ is unmatched at the end, some woman $w$ displaces him during
this iteration.
The pair $(m,w)$ was not frozen, since frozen pairs are permanent.
Hence $d(m)\leq R d(w)$.
We charge this contribution of $d(m)$ to the displacement by $w$.

Each time a woman displaces her partner, she moves to a strictly better
quantile.
Thus, each woman $w$ causes less than $k$ displacements, and the
total charge of this type is less than
\[
    R\sum_{w\in X} k d(w)=kR|E|.
\]

Combining the two types of charges gives
\[
    \sum_{t=1}^L P_t
    < k|E|+kR|E|
    = k(R+1)|E|.
    \qedhere
\]
\end{proof}

Since the output iteration $J$ is chosen uniformly at random, the preceding
lemma implies that its expected unmatched degree is small.

\begin{lemma}[Degree bound for unmatched men at the output iteration]
\label{lem:random-bad-mass} $\E[P_J]\leq \frac{\eps}{4}|E|$.
\end{lemma}

\begin{proof}
Since $J$ is uniform over $\{1,\ldots,L\}$ and independent of the
proposal process,
\[
    \E[P_J]
    =
    \frac{1}{L}\sum_{t=1}^{L}\E[P_t]
    =
    \frac{1}{L}\E\left[\sum_{t=1}^{L}P_t\right].
\]
By \Cref{lem:bad-degree-charging} and the definition of $L$,
\[
    \E[P_J]
    \leq
    \frac{k(R+1)}{L}|E|
    \leq
    \frac{\eps}{4}|E|.
    \qedhere
\]
\end{proof}

We now translate the preceding bounds into a bound on the number of
blocking pairs.
Following Ostrovsky and Rosenbaum~\cite{ostrovsky2015fast}, we distinguish
blocking pairs according to how much both endpoints improve.
For a matching $M$, define
\[
    \Delta_v(u;M)=P_v(p_M(v))-P_v(u).
\]
A blocking pair $(m,w)$ is \emph{far} if
\[
    \Delta_m(w;M)>\frac{d(m)}{k}
    \qquad\text{and}\qquad
    \Delta_w(m;M)>\frac{d(w)}{k};
\]
otherwise, it is \emph{near}.

The following is a direct adaptation of the counting argument in
\cite[Lemma~4]{ostrovsky2015fast}.

\begin{lemma}[Near blocking pairs]
\label{lem:near-blocking}
Every matching $M$ has at most $2|E|/k$ near blocking pairs.
\end{lemma}

\begin{proof}
For any vertex $v$, at most $d(v)/k$ neighbors $u$ satisfy
$0<\Delta_v(u;M)\leq d(v)/k$.
For every near blocking pair, this inequality holds for at least one of its
endpoints.
Therefore, the number of near blocking pairs is at most
\[
    \sum_{m\in Y}\frac{d(m)}{k}
    +
    \sum_{w\in X}\frac{d(w)}{k}
    =
    \frac{2|E|}{k}.
    \qedhere
\]
\end{proof}

\begin{lemma}[Far blocking pairs]
\label{lem:far-blocking}
For every $t\in\{1,\ldots,L\}$, every far blocking pair of $M_t$ belongs
to $F\cup Z$ or is incident to a man in $B_t$.
Consequently,
\begin{equation}
\label{eq:blocking-decomposition}
    \bp(M_t)
    \leq
    \frac{2|E|}{k}+|F|+|Z|+P_t.
\end{equation}
\end{lemma}

\begin{proof}
Let $(m,w)$ be a far blocking pair of $M_t$ with
$(m,w)\notin F\cup Z$.
We prove that $m\in B_t$.

Suppose, for a contradiction, that $m\notin B_t$.
We first claim that $(m,w)$ must have been deleted by the end of
iteration $t$.
There are two cases.

First, suppose that $m$ is matched in $M_t$.
Since $(m,w)$ is far,
\[
    P_m(p_{M_t}(m))-P_m(w)>\frac{d(m)}{k}.
\]
Two neighbors in the same quantile of $m$ differ in rank by at most
\[
    s_m-1<\frac{d(m)}{k}.
\]
Hence, $w$ belongs to a strictly better quantile of $m$ than
$p_{M_t}(m)$.

Consider the iteration in which $m$ obtained his current partner
$p_{M_t}(m)$.
At the beginning of that iteration, $m$ selected his first original
quantile containing a live edge, and his current partner belongs to that
quantile.
Since the quantile containing $w$ is strictly better, it contained no
live edge at that time.
In particular, $(m,w)$ had already been deleted.

Second, suppose that $m$ is unmatched in $M_t$.
Since $m\notin B_t$, by the definition of $B_t$ he has no live incident
edge after iteration $t$.
Thus, $(m,w)$ has again been deleted by the end of iteration $t$.

This proves the claim.
Since $(m,w)\notin Z$, it was not deleted as a residual edge.
Moreover, $(m,w)\notin F$ means that $w$ never forms a frozen pair, so
$(m,w)$ was not deleted by the degree guard.
Therefore, $(m,w)$ must have been rejected by $w$.

When $w$ rejected $m$, she obtained a partner in the same quantile as
$m$ or in a better quantile.
Thereafter, whenever $w$ changes partners, she moves to a strictly better
quantile.
Hence, in $M_t$, either $w$ prefers her partner to $m$, contradicting
that $(m,w)$ is blocking, or $m$ and $p_{M_t}(w)$ belong to the same
quantile.
In the latter case,
\[
    \Delta_w(m;M_t)
    \leq s_w-1
    <\frac{d(w)}{k},
\]
contradicting that $(m,w)$ is far.

Thus, the assumption $m\notin B_t$ leads to a contradiction.
Therefore, every far blocking pair outside $F\cup Z$ is incident to a
man in $B_t$, that is, an unmatched man with at least one live incident
edge after iteration $t$.

By \Cref{lem:near-blocking}, there are at most $2|E|/k$ near blocking
pairs.
There are at most $|F|+|Z|$ far blocking pairs in $F\cup Z$, while the
number of edges incident to men in $B_t$ is at most
\(\sum_{m\in B_t}d(m)=P_t\). 
This proves \eqref{eq:blocking-decomposition}.
\end{proof}

We can now prove the main theorem.

\maincongest*

\begin{proof}
Apply \eqref{eq:blocking-decomposition} to $M_J$ and take expectations.
By
\Cref{lem:frozen-edges,lem:residual-edges,lem:random-bad-mass},
\begin{align*}
    \E[\bp(M_J)]
    &\leq
    \frac{2|E|}{k}
    +\frac{|E|}{R}
    +\E[|Z|]
    +\E[P_J]\\
    &\leq
    \frac{\eps}{4}|E|
    +\frac{\eps}{4}|E|
    +\frac{\eps}{4}|E|
    +\frac{\eps}{4}|E|\\
    &=\eps|E|,
\end{align*}
where we use $k\geq 8/\eps$ and $R=4/\eps$.

For the round complexity, by \eqref{eq:parameters},
\[
    k=\Theta(\eps^{-1}),
    \qquad
    R=\Theta(\eps^{-1}),
    \qquad
    L=\Theta(\eps^{-3}),
    \qquad
    \rho=\Theta(\eps^5).
\]

To evaluate the degree guard \eqref{eq:freeze-condition}, each vertex
needs to know the original degrees of its neighbors.
Thus, at the beginning of the algorithm, every vertex sends its original
degree to all of its neighbors.
This takes one $\congest$ round using $O(\log n)$-bit messages.
Afterward, all proposal, acceptance, rejection, and matching-state
messages use only $O(1)$ bits per edge per communication round.

A call to $\AMM(H,\rho)$ takes
$\lceil\log_2(1/\rho)\rceil=O(\log(1/\eps))$ steps, each of which can be
implemented in $O(1)$ $\congest$ rounds.
All other operations in a call to $\GPR$ also require only $O(1)$
$\congest$ rounds.
Thus, each call to $\GPR$ takes $O(\log(1/\eps))$ $\congest$ rounds.

There are $L$ calls to $\GQM$, each consisting of $k$ calls to
$\GPR$, so the total round complexity is
\[
    O\left(Lk\log(1/\eps)\right)
    =
    O\left(\frac{\log(1/\eps)}{\eps^4}\right).
\]

Finally, since $L$ is a power of two, a uniformly random index
$J\in\{1,\ldots,L\}$ can be selected using exactly
\(
    \log_2 L=O(\log(1/\eps))
\) shared random bits.
\end{proof}

\section{Removing Shared Randomness}
\label{sec:removing-shared-randomness}

In this section, we remove the assumption of pre-shared randomness using
a low-diameter decomposition.
We partition the graph into low-diameter clusters, generate the required
shared randomness independently within each cluster, and run the
algorithm of \Cref{sec:constant-round-congest} in parallel on the
resulting induced subgraphs.

We use the low-diameter decomposition algorithm of Miller, Peng, and
Xu~\cite{miller2013parallel}.
Although originally presented in the parallel setting, its implementation
in the $\congest$ model is standard; see, for example,
Forster, Gr\"osbacher, and de Vos~\cite{forster2022improved}.
In the subsequent discussion, the \emph{strong diameter} of a vertex
subset $C$ is the diameter of the subgraph induced by $C$.

\begin{lemma}[Low-diameter decomposition~\cite{miller2013parallel,forster2022improved}]
\label{lem:congest-ldd}
For every unweighted graph $G=(V,E)$ and every $0<\beta\leq 1$, there is
a randomized $\congest$ algorithm that, in
$O\left(\frac{\log n}{\beta}\right)$ rounds, partitions $V$ into
clusters of strong diameter
$O\left(\frac{\log n}{\beta}\right)$.
Moreover, every edge has its endpoints in different clusters with
probability at most $\beta$.
\end{lemma}

We now prove \Cref{cor:decomposition}.

\decomposition*

\begin{proof}
Set $\beta=\eta=\eps/2$ and apply \Cref{lem:congest-ldd} with parameter
$\beta$.
Let $\mathcal{C}=\{C_1,\ldots,C_q\}$ be the resulting partition.
For each $i$, let $G_i=G[C_i]$ be the subgraph induced by $C_i$, and let
$E_i=E(G_i)$ be its edge set.
Let $E_{\mathrm{cut}}$ be the set of edges joining different clusters.
By \Cref{lem:congest-ldd} and linearity of expectation,
\[
    \E[|E_{\mathrm{cut}}|]\leq \beta|E|.
\]

We next generate the shared randomness required within each cluster.
Since every cluster has strong diameter
$O\left(\frac{\log n}{\beta}\right)$, its vertices can elect a leader
and construct a rooted spanning tree of this depth within
$O\left(\frac{\log n}{\beta}\right)$ rounds~\cite{peleg2000distributed}.
The leader of each cluster privately chooses the uniformly random output
index required by $\DGASM$ with accuracy parameter $\eta$ and broadcasts
it within the cluster.
As the index is represented using $b=O(\log(1/\eta))$ bits,  by
pipelining $O(\log n)$ bits per edge per round, the low-diameter
decomposition, leader election, and broadcast of the cluster-wise shared
randomness together take
\[
    O\left(
        \frac{\log n}{\beta}
        +
        \frac{b}{\log n}
    \right)
    =
    O\left(
        \frac{\log n}{\eps}
        +
        \frac{\log(1/\eps)}{\log n}
    \right)
\]
rounds.
All clusters perform these operations in parallel.

We then run $\DGASM$ independently on each $G_i$ with accuracy parameter
$\eta$, using degrees, preference ranks, and quantiles with respect to
$G_i$.
Inter-cluster edges are ignored, and all clusters execute the algorithm
in parallel.
If $M_i$ is the matching returned on $G_i$, we output
$M=\bigcup_{i=1}^q M_i$.

Every blocking pair of $M$ is either an inter-cluster edge or is contained
in some $G_i$.
Hence
\[
    \bp_G(M)
    \leq |E_{\mathrm{cut}}|
    +\sum_{i=1}^q \bp_{G_i}(M_i).
\]
Conditioned on the decomposition $\mathcal C$, \Cref{thm:main} gives
$\E[\bp_{G_i}(M_i)\mid\mathcal C]\leq\eta|E_i|$ for every $i$.
Since $\sum_i|E_i|\leq|E|$, taking expectations yields
\[
    \E[\bp_G(M)]
    \leq \beta|E|+\eta|E|
    =\eps|E|.
\]

The executions of $\DGASM$ run in parallel and take
$O\left(\frac{\log(1/\eps)}{\eps^4}\right)$ rounds.
Therefore, the total round complexity is
\[
    O\left(
        \frac{\log(1/\eps)}{\eps^4}
        +
        \frac{\log n}{\eps}
    \right).
    \qedhere
\]
\end{proof}

\section{Massively Parallel Computation}
\label{sec:mpc}

In this section, we show that the algorithm of
\Cref{sec:constant-round-congest} can be implemented in the
fully-scalable $\mpc$ model without increasing its asymptotic round
complexity.
The combinatorial algorithm is unchanged.
The main issue is that the edges incident to a high-degree vertex may be
distributed across many machines, so operations such as finding the best
proposal quantile or updating all incident edges cannot be performed
locally on a single machine.
We show that all such operations can be implemented using standard $\mpc$
sorting and prefix-sum primitives.

\paragraph{Basic primitives.}
Recall that each machine has $S= \Theta(n^\delta)$ words of local memory, for an
arbitrarily small fixed constant $\delta>0$.
We use $O((n+|E|)/S)$ machines, and hence $O(n+|E|)$ total memory.

The machines are indexed by their machine IDs $1,2,3,\ldots$, and the records stored on each machine are indexed by their local positions. This induces a global ordering of all records, ordered first by machine ID and then by local position.
A \emph{sorting} operation takes a designated collection of records with
$O(1)$-word keys and redistributes them so that their keys appear in
nondecreasing order in the global ordering.
A \emph{prefix-sum} operation takes values $x_1,\ldots,x_N$ stored in
records ordered according to the global ordering and computes \(\sum_{j=1}^i x_j\) for every
$i$.
Records not participating in an operation can simply be ignored.

Sorting and prefix sums on $N$ records of $O(1)$ words can be implemented in
$O(\log_S N)$ rounds~\cite{goodrich2011sorting}.
In our setting, $N=O(n+|E|)$ and $\log_S N=O(1)$.

\begin{lemma}[Basic primitives]
\label{lem:mpc-primitives}
Using $O(n+|E|)$ total memory, each of the following operations can be
implemented in $O(1)$ $\mpc$ rounds:
\begin{itemize}
    \item Sorting $O(n+|E|)$ records, each consisting of $O(1)$ words.
    \item Computing prefix sums of $O(n+|E|)$ $O(1)$-word values.
    \item Broadcasting an $O(1)$-word value from one machine to all
    machines.
\end{itemize}
\end{lemma}

\begin{proof}
The first two operations follow from
\citet{goodrich2011sorting}.
For the third, first send the value to the machine with the smallest
machine ID.
Each machine then creates one designated record: the first machine
stores the value in its record, while every other machine stores zero.
A prefix-sum computation on these records copies the value to every
machine.
\end{proof}

We will also use these operations independently on many disjoint blocks
of records in parallel.
In particular, after sorting records by vertex identifier, all records
associated with the same vertex form one consecutive block.
Prefix sums and broadcasts can then be carried out separately within
each block, simultaneously for all vertices.
This allows us, for example, to count the records associated with each
vertex, assign them consecutive indices, or propagate an $O(1)$-word
value associated with the vertex to every record in its block.

\paragraph{Implementing the algorithm.}
We maintain one $O(1)$-word record for every vertex and every original
edge.
A vertex record stores its original degree, current partner, active
status and active quantile when applicable.
An edge record stores its two preference ranks, its two quantile indices,
and whether it is live.
We use only a constant number of additional temporary fields per record
throughout the execution.
We first explain how the original degrees and quantile indices stored in
these records are computed.

\paragraph{Preprocessing step 1: Degree and quantile computation.}
For every edge $e=\{u,v\}$, we create two temporary incidence records,
one associated with $u$ and one with $v$.
The incidence record associated with $v$ stores the edge identifier
$e$ and the preference rank of the other endpoint in $v$'s preference
list.
We sort the incidence records, together with the vertex records, by
vertex identifier.

For each vertex $v$, its incidence records now form one consecutive
block.
A prefix sum of ones over this block assigns the values
$1,\ldots,d(v)$ to its incidence records; in particular, the last
incidence record learns $d(v)$.
We then broadcast this value within the block and store it in the vertex
record of $v$.
Consequently, every incidence record associated with $v$ knows both
$d(v)$ and the corresponding preference rank, and can locally compute
\[
    s_v=\left\lceil\frac{d(v)}{k}\right\rceil
    \qquad\text{and}\qquad
    q_v(u)=\left\lceil\frac{P_v(u)}{s_v}\right\rceil.
\]
Recall that $s_v$ is the size parameter used to partition $v$'s
preference list into quantiles, and $q_v(u)$ is the index of the
quantile containing $u$.

Finally, we sort the incidence records by edge identifier and store the
two resulting quantile indices in the corresponding edge record.
Thus, all original degrees and quantile indices are computed in
$O(1)$ $\mpc$ rounds using $O(n+|E|)$ total memory.

We use the same mechanism throughout the algorithm.
Whenever an $O(1)$-word state associated with a vertex is needed by all
of its incident edges, we sort the corresponding records by vertex
identifier and perform a broadcast within each vertex block.
This takes $O(1)$ $\mpc$ rounds even when a block spans many machines,
so the incident edges of a high-degree vertex never need to fit on a
single machine.

\paragraph{Preprocessing step 2: Shared random output index.}
The shared output index $J$ can also be generated in $O(1)$ $\mpc$
rounds.
One designated machine samples $J$ uniformly from $\{1,\ldots,L\}$ and
broadcasts it to all machines.
Since $1/\eps\leq n^2$, the value $J$ can be represented using
$O(\log n)$ bits and hence occupies $O(1)$ words.

We next describe the $\mpc$ implementation of $\GPR$.

\begin{lemma}[Implementation of a guarded proposal round]
\label{lem:mpc-gpr}
One call to $\GPR$ can be implemented in
$O(\log(1/\rho))$ $\mpc$ rounds using $O(n+|E|)$ total memory.
\end{lemma}

\begin{proof}
The proposal and acceptance steps of $\GPR$ can be implemented in
$O(1)$ $\mpc$ rounds using the standard primitives described above.
In particular, the proposals are sorted by their female endpoints and
quantile indices, allowing every woman to identify her best quantile
containing a proposal and hence the accepted proposal graph $H$.

It remains to implement $\AMM(H,\rho)$.
In each step, for every free man we index his free neighbors, sample a
uniformly random index, and send a proposal to the corresponding woman.
The proposals are then grouped by their female endpoints, and every
woman receiving at least one proposal selects one of them according to
a fixed tie-breaking rule.
All of these operations can be performed in $O(1)$ $\mpc$ rounds.
Thus, since $\AMM$ performs
$\lceil\log_2(1/\rho)\rceil$ steps, it takes
$O(\log(1/\rho))$ $\mpc$ rounds.

After $\AMM$ terminates, all remaining operations of $\GPR$ consist of
identifying residual edges and performing deterministic updates to
vertex and edge states.
These can be implemented in $O(1)$ $\mpc$ rounds using the standard
primitives above.
Thus, the entire call takes $O(\log(1/\rho))$ rounds and faithfully
simulates $\GPR$.
Throughout the procedure, only a constant number of $O(1)$-word records and
temporary fields are maintained per original vertex and edge, so the
total memory is $O(n+|E|)$.
\end{proof}

We have now described all the ingredients needed to implement $\DGASM$
in the fully-scalable $\mpc$ model.
We can therefore prove \Cref{cor:mpc}.

\massivelyparallel*

\begin{proof}
At the beginning of each call to $\GQM$, we sort the live edges incident
to each unmatched man by their quantile indices.
This identifies his first quantile containing a live edge in $O(1)$
$\mpc$ rounds, which is then stored as his active quantile.
The procedure subsequently executes $k$ calls to $\GPR$.
The algorithm executes $\GQM$ for $J$ iterations and then stops,
returning the current matching.

By the implementation described above, the $\mpc$ algorithm simulates
the first $J$ iterations of $\DGASM$.
Since $J$ is chosen independently of the proposal process, the returned
matching has the same distribution as $M_J$ in the execution analyzed
in \Cref{thm:main}.
Hence, the guarantee
$\E[\bp(M)]\leq\eps|E|$ follows directly from the proof of
\Cref{thm:main}.

For the round complexity, recall from \eqref{eq:parameters} that
$k=\Theta(\eps^{-1})$, $L=\Theta(\eps^{-3})$, and
$\rho=\Theta(\eps^5)$.
Since $J\leq L$, there are at most $Lk$ calls to $\GPR$, each taking
$O(\log(1/\rho))=O(\log(1/\eps))$ $\mpc$ rounds by
\Cref{lem:mpc-gpr}.
All remaining operations take only $O(1)$ rounds per call to $\GQM$,
in addition to the $O(1)$ preprocessing.
Hence, the total round complexity is
\[
    O\left(
        Lk\log(1/\rho)
    \right)
    =
    O\left(
        \frac{\log(1/\eps)}{\eps^4}
    \right).
\]

Throughout the execution, we maintain only a constant number of
$O(1)$-word records per original vertex and edge.
Thus, the total memory is $O(n+|E|)$ words, with at most $n^\delta$
words stored on each machine.
\end{proof}

\section{Conclusions and Open Questions}
\label{sec:conclusions}

We showed that, for every constant $\epsilon>0$, a
$(1-\epsilon)$-stable matching can be computed in $O(1)$ rounds with a
small amount of shared randomness (\Cref{thm:main}), and in
$O(\log n)$ rounds without pre-shared randomness
(\Cref{cor:decomposition}).
This improves the previous randomized upper bound of
$O(\log^2 n)$ due to Ostrovsky and Rosenbaum~\cite{ostrovsky2015fast}.
Our results leave several natural directions for further investigation.

\paragraph{Is shared randomness necessary?}
The most immediate question is whether the gap between the models with
and without shared randomness can be closed.
Shared randomness is known to increase the power of distributed graph
algorithms: \citet{balliu2024shared} established an exponential
separation between the two models for the class of locally checkable
labeling problems.
To the best of our knowledge, however, no such separation is known for
a natural distributed graph problem.
Can constant round complexity be achieved without shared
randomness for almost stable matching?

\paragraph{What is the optimal deterministic complexity?}
The optimal deterministic round complexity of almost stable matching
remains open.
Ostrovsky and Rosenbaum~\cite{ostrovsky2015fast} give the bound
\(
    O\left(
        \frac{\dtmax\log n}{\epsilon^3}
    \right)
\),
where $\dtmax$ denotes the deterministic round complexity of maximal
matching.
It would be interesting either to improve the dependence on $n$ or to
show that some nonconstant dependence is unavoidable.

Proving such a lower bound appears challenging. For the related problem of approximating maximum independent set,
Lenzen and Wattenhofer~\cite{lenzen2008leveraging} proved an
$\Omega(\log^* n)$ deterministic lower bound even on cycles.
The same lower bound applies to approximating maximum matching, since the two problems
are locally equivalent on cycles. 
Their construction, however, cannot be adapted directly to almost stable
matching: for bounded-degree graphs and constant $\epsilon$,
\citet{floreen2010almost} give an $O(1)$-round algorithm for
$(1-\epsilon)$-stable matching.
Thus, any $\omega(1)$ deterministic lower bound for almost stable
matching must exploit instances with unbounded degree.

\paragraph{What is the optimal dependence on $\epsilon$?}
Finally, there remains a polynomial gap between our
\(
    O\left(
        \frac{\log(1/\epsilon)}{\epsilon^4}
    \right)
\)
upper bound and the $\Omega(1/\epsilon)$ lower bound for general graphs.
Can the upper bound be improved to nearly linear in $1/\epsilon$, or is
a stronger lower bound possible?
For bounded-degree graphs, the dependence on $\epsilon$ is already tight
up to constant factors, by the
$O\left(1/\epsilon\right)$-round algorithm of
\citet{floreen2010almost}.

\section*{AI Disclosure}
This project was originally developed during 2024--2025, but the authors
abandoned the manuscript after discovering a critical flaw in the proof
in January 2025 that they were unable to repair at the time.
The project was revived in 2026 with the assistance of OpenAI's ChatGPT
(GPT-5.6 Sol).
ChatGPT proposed the degree-guarded freezing rule and used it to repair
the flawed proof.
ChatGPT also assisted extensively in writing the manuscript, including
drafting and revising substantial portions of the text throughout the
paper.
All AI-generated content was checked and revised by the authors before
being incorporated into the paper.
The authors take full responsibility for the correctness of the paper.

\printbibliography
\appendix

\section{Lower Bound}
\label{sec:LB}

In this section, we recall the $\Omega(1/\epsilon)$ lower bound for
almost stable matching.
It follows directly from the $\Omega(n)$-round lower bound of
Flor\'{e}en, Kaski, Polishchuk, and Suomela~\cite{floreen2010almost}
for exact stable matching.
Indeed, their hard instance can be chosen to have
$n=\Theta(1/\epsilon)$ vertices and fewer than $1/\epsilon$ edges.
Thus, every $(1-\epsilon)$-stable matching on this instance has no
blocking pair and is therefore exactly stable.
For completeness, we give the argument below.

The lower bound is based on a reduction to the \emph{mailing
problem}~\cite{peleg2000near}.
In this problem, we are given a graph $G=(V,E)$ with two designated
vertices: a \emph{sender} $s$ and a \emph{receiver} $r$.
The sender is given a bit $b\in\{0,1\}$, and the goal is for the receiver
to output $b$.
The following standard observation also holds in the presence of shared
randomness.

\begin{lemma}
\label{lem:lower-bound}
In the $\local$ model with shared randomness, any algorithm that solves
an instance $(G,s,r)$ of the mailing problem with probability greater
than $1/2$ requires at least $d$ rounds, where
$d=\dist(s,r)$.
\end{lemma}

\begin{proof}
Suppose that the algorithm terminates in fewer than $d$ rounds.
Then no information about the input bit $b$ can reach $r$, so the
distribution of the output of $r$ is independent of $b$.
Let $x\in\{0,1\}$ be such that the probability that $r$ outputs $x$ is
at most $1/2$.
If $b=x$, the algorithm therefore succeeds with probability at most
$1/2$.
Hence, any algorithm with success probability greater than $1/2$ must
take at least $d$ rounds.
\end{proof}

We now describe the hard stable marriage instance.
The construction is due to Flor\'{e}en, Kaski, Polishchuk, and
Suomela~\cite{floreen2010almost}.
It was also used as a building block in the
$\congest$ lower bound of Kipnis and
Patt-Shamir~\cite{kipnis2009note}, and our presentation follows their
approach.

It suffices to consider sufficiently small $\epsilon$.
Let $H$ be the path
\[
    \langle v_0,v_1,\ldots,v_{k-1},v_k\rangle,
    \qquad
    k=\left\lceil\frac{1}{\epsilon}\right\rceil-1.
\]
We view $H$ as a bipartite stable marriage instance using the natural
bipartition according to the parity of the vertex indices.
Set $s=v_1$ and $r=v_k$.
The sender $s$ is given a bit $b\in\{0,1\}$, which determines its
preference list as follows.
For each $i\in\{2,\ldots,k-1\}$, vertex $v_i$ prefers $v_{i-1}$ over
$v_{i+1}$.
At $v_1$, if $b=0$, then $v_1$ prefers $v_0$ over $v_2$, whereas if
$b=1$, then $v_1$ prefers $v_2$ over $v_0$.
The endpoints $v_0$ and $v_k$ each have only one neighbor.

\begin{lemma}
\label{lem:lower-bound-first}
Let $M$ be any stable matching of $H$.
For every $i\in\{0,\ldots,k-2\}$, $M$ contains exactly one of
$\{v_i,v_{i+1}\}$ and $\{v_{i+1},v_{i+2}\}$.
\end{lemma}

\begin{proof}
Since $M$ is a matching, it contains at most one of the two edges.
Suppose that it contains neither.
Then $v_{i+1}$ is unmatched.
Moreover, $v_{i+2}$ ranks $v_{i+1}$ first among its neighbors, while
$v_{i+1}$ prefers being matched to $v_{i+2}$ over being unmatched.
Thus, $\{v_{i+1},v_{i+2}\}$ is a blocking pair, contradicting the
stability of $M$.
\end{proof}

\begin{lemma}
\label{lem:unique-stable-matching}
The graph $H$ admits a unique stable matching
\[
M=
\begin{cases}
    \{\{v_0,v_1\},\{v_2,v_3\},\ldots\},
        & \text{if $b=0$,}\\
    \{\{v_1,v_2\},\{v_3,v_4\},\ldots\},
        & \text{if $b=1$.}
\end{cases}
\]
\end{lemma}

\begin{proof}
Let $M$ be any stable matching.
By \Cref{lem:lower-bound-first}, it suffices to determine whether
$\{v_0,v_1\}$ belongs to $M$.

Suppose first that $b=0$.
If $\{v_0,v_1\}\notin M$, then $v_0$ is unmatched, while $v_1$ prefers
$v_0$ over $v_2$ and over being unmatched.
Hence, $\{v_0,v_1\}$ is a blocking pair.
Therefore, $\{v_0,v_1\}\in M$.

Now suppose that $b=1$.
If $\{v_0,v_1\}\in M$, then $\{v_1,v_2\}\notin M$.
However, $v_1$ prefers $v_2$ over $v_0$, and $v_2$ ranks $v_1$ first.
Thus, $\{v_1,v_2\}$ is a blocking pair.
Therefore, $\{v_0,v_1\}\notin M$.

The remainder of the matching is then uniquely determined by
\Cref{lem:lower-bound-first}.
\end{proof}

We can now prove the lower bound.

\lowerbound*

\begin{proof}
The graph $H$ has
\[
    |E(H)|
    =
    k
    =
    \left\lceil\frac{1}{\epsilon}\right\rceil-1
    <
    \frac{1}{\epsilon}.
\]
Hence, any $(1-\epsilon)$-stable matching of $H$ has fewer than one
blocking pair and therefore has no blocking pair.
Thus, every $(1-\epsilon)$-stable matching of $H$ is an exact stable
matching.

By \Cref{lem:unique-stable-matching}, the receiver $r=v_k$ can determine
the input bit $b$ from whether its incident edge
$\{v_{k-1},v_k\}$ belongs to the matching; which of the two possibilities
corresponds to $b=0$ is determined by the parity of $k$ and is known to
$r$.
Therefore, any algorithm that computes a $(1-\epsilon)$-stable matching
of $H$ with probability greater than $1/2$ yields an algorithm for the
mailing problem on $(H,s,r)$ with the same success probability.

By \Cref{lem:lower-bound}, such an algorithm requires at least
\[
    \dist(s,r)
    =
    k-1
    =
    \left\lceil\frac{1}{\epsilon}\right\rceil-2
    =
    \Omega\left(\frac{1}{\epsilon}\right)
\]
rounds.
This proves the theorem.
\end{proof}

\end{document}